\documentclass[a4paper,USenglish]{lipics-v2021}
\usepackage{graphicx,bm,xspace}
\usepackage{amsmath}
\usepackage{amssymb}
\usepackage{amsthm}
\usepackage{color}
\usepackage{mathtools}
\usepackage{physics}
\usepackage{hyperref}
\usepackage{tikzit}

\tikzstyle{box}=[shape=rectangle, text height=1.5ex, text depth=0.25ex, yshift=0.5mm, fill=white, draw=black, minimum height=5mm, yshift=-0.5mm, minimum width=5mm, font={\small}]
\tikzstyle{Z dot}=[inner sep=0mm, minimum size=2mm, shape=circle, draw=black, fill={rgb,255: red,216; green,248; blue,216}, tikzit fill={rgb,255: red,216; green,248; blue,216}]
\tikzstyle{Z phase dot}=[minimum size=4.75mm, font={\footnotesize}, shape=rectangle, rounded corners=1.9mm, inner sep=0.1mm, outer sep=-2mm, scale=0.8, tikzit shape=circle, draw=black, fill={rgb,255: red,216; green,248; blue,216}, tikzit draw=blue, tikzit fill={rgb,255: red,216; green,248; blue,216}]
\tikzstyle{X dot}=[Z dot, shape=circle, draw=black, fill={rgb,255: red,232; green,165; blue,165}, tikzit fill={rgb,255: red,232; green,165; blue,165}]
\tikzstyle{X phase dot}=[Z phase dot, tikzit shape=circle, tikzit fill={rgb,255: red,232; green,165; blue,165}, fill={rgb,255: red,232; green,165; blue,165}, font={\footnotesize}, tikzit draw=blue]
\tikzstyle{XD dot}=[shape=XDdot, inner sep=2pt, draw=black]
\tikzstyle{XD phase dot}=[shape=XDdotphase, minimum size=4.75mm, font={\footnotesize}, inner sep=.1mm, outer sep=0mm, scale=0.8, tikzit shape=circle, rounded corners=1.9mm, draw=black]
\tikzstyle{zn}=[shape=zn, tikzit draw=black, draw=black, inner sep=2pt]
\tikzstyle{hadamard}=[fill={rgb,255: red,140; green,220; blue,248}, draw=black, shape=rectangle, inner sep=0.6mm, minimum height=1.5mm, minimum width=1.5mm, xslant=0.5]
\tikzstyle{hz}=[hadamard, fill={rgb,255: red,216; green,248; blue,216}, shape=rectangle, tikzit fill={rgb,255: red,216; green,248; blue,216}, minimum height=2 mm, minimum width=1.25 mm, tikzit draw=black]
\tikzstyle{hx}=[hadamard, fill={rgb,255: red,232; green,165; blue,165}, shape=rectangle, tikzit fill={rgb,255: red,232; green,165; blue,165}, minimum height=2 mm, minimum width=1.25 mm, tikzit draw=black]
\tikzstyle{vertex}=[inner sep=0mm, minimum size=1mm, shape=circle, draw=black, fill=black]
\tikzstyle{vertex set}=[inner sep=0mm, minimum size=1mm, shape=circle, draw=black, fill=white, font={\footnotesize\boldmath}]
\tikzstyle{meter}=[draw, fill=white, minimum width=2em, minimum height=1.5em, rectangle, path picture={\draw ([shift={(.1,.24)}]path picture bounding box.south west) to[bend left=50] ([shift={(-.1,.24)}]path picture bounding box.south east);\draw[-{Latex[scale=0.6]}] ([shift={(0,.1)}]path picture bounding box.south) -- ([shift={(.3,-.1)}]path picture bounding box.north);}, tikzit shape=rectangle]
\tikzstyle{white dot}=[Z dot]
\tikzstyle{gray dot}=[X dot]
\tikzstyle{white phase dot}=[Z phase dot]
\tikzstyle{gray phase dot}=[X phase dot]
\tikzstyle{red ket}=[fill={rgb,255: red,232; green,165; blue,165}, draw=black, shape=isosceles triangle, tikzit fill={rgb,255: red,232; green,165; blue,165}, tikzit draw=black, inner sep=0 mm, outer sep=2 mm]
\tikzstyle{tiny none}=[none, font={\tiny}]
\tikzstyle{green ket}=[fill={rgb,255: red,216; green,248; blue,216}, draw=black, shape=isosceles triangle, tikzit fill={rgb,255: red,216; green,248; blue,216}, tikzit draw=black, inner sep=0 mm, outer sep=2 mm]
\tikzstyle{filament}=[hadamard, fill=yellow, draw=none, minimum height=0.01mm]
\tikzstyle{unit circle}=[shape=circle, minimum size=42.5 mm, fill=none, draw={rgb,255: red,223; green,223; blue,223}, tikzit draw={rgb,255: red,223; green,223; blue,223}]
\tikzstyle{new style 0}=[shape=ellipse, minimum height=10 mm, minimum width=100 mm, fill=white, draw=black]
\tikzstyle{gate}=[box, minimum height=10mm, minimum width=10mm]
\tikzstyle{2 control}=[vertex set, draw=blue, inner sep=0.5pt]
\tikzstyle{1 control}=[2 control, draw=red]
\tikzstyle{0 control}=[2 control, draw=black]
\tikzstyle{small dot}=[vertex, minimum size=1 mm, draw=black, tikzit draw=black, tikzit fill=black, tikzit shape=circle]
\tikzstyle{tallbox}=[box, minimum height=12mm]
\tikzstyle{targ}=[vertex set, minimum size=0.5mm, inner sep=-0.5mm, tikzit shape=circle, shape=circle, tikzit draw=black]
\tikzstyle{hadamardbox}=[hadamard, xslant=0]

\tikzstyle{directedarrow}=[draw={rgb,255: red,223; green,223; blue,223}, ->, tikzit draw={rgb,255: red,223; green,223; blue,223}, line width=1 pt]
\tikzstyle{simple}=[-]
\tikzstyle{hadamard edge}=[-, color={rgb,255: red,0; green,100; blue,248}, dashed, dash pattern=on 2pt off 0.7pt, tikzit draw={rgb,255: red,0; green,100; blue,248}]
\tikzstyle{brace edge}=[-, tikzit draw=blue, decorate, decoration={brace,amplitude=1mm,raise=-1mm}]
\tikzstyle{gray}=[-, draw={rgb,255: red,223; green,223; blue,223}, line width=1 pt]
\tikzstyle{arrow}=[<-, draw={rgb,255: red,128; green,128; blue,128}]
\tikzstyle{double-arrow}=[draw={rgb,255: red,128; green,128; blue,128}, <->]
\tikzstyle{dashed edge}=[-, dashed, dash pattern=on 2pt off 0.5pt, draw=black]
\tikzstyle{diredge}=[->]
\tikzstyle{double edge}=[-, double, shorten <=-1mm, shorten >=-1mm, double distance=2pt]
\tikzstyle{thin}=[-, line width=0.05mm]
\tikzstyle{thin gray}=[-, draw={rgb,255: red,223; green,223; blue,223}, line width=0.05mm]
\tikzstyle{less thin}=[-, line width=0.1mm]
\tikzstyle{dashed gray edge}=[-, dashed edge, draw={rgb,255: red,128; green,128; blue,128}]
\tikzstyle{light right directed arrow}=[->, directedarrow, draw={rgb,255: red,223; green,223; blue,223}, line width=0.2mm]
\tikzstyle{diredge0.3}=[->, line width=0.3 mm]
\tikzstyle{less thin gray}=[-, draw={rgb,255: red,223; green,223; blue,223}]
\tikzstyle{dashed thin purple}=[-, dashed, line width=0.1mm, draw={rgb,255: red,128; green,106; blue,219}]
\tikzstyle{hadamardedge}=[-, color={rgb,255: red,100; green,200; blue,248}, dashed, dash pattern=on 2pt off 0.7pt, tikzit draw={rgb,255: red,120; green,220; blue,248}]

\input{zx.tikzdefs}

\DeclareUnicodeCharacter{2212}{-}                                                                             
\definecolor{Red}{rgb}{0.6,0.1,0.2}

\title{Qutrit metaplectic gates are a subset of Clifford+T}
\titlerunning{Qutrit metaplectic gates are a subset of Clifford+T}

\author{Andrew Glaudell}{Booz Allen Hamilton, United States \and Department of Mathematics, George Mason University, United States \and Quantum Science and Engineering Center, George Mason University, United States}{andrewglaudell@gmail.com}{https://orcid.org/0000-0001-9824-5804}{}
\author{Neil J. Ross}{Department of Mathematics and Statistics, Dalhousie University, Canada \and \url{https://www.mathstat.dal.ca/~neilr/}}{neil.jr.ross@dal.ca}{https://orcid.org/0000-0003-0941-4333}{NSERC Discovery Grant}
\author{John van de Wetering}{Radboud University Nijmegen, The Netherlands \and University of Oxford, United Kingdom \and \url{https://vdwetering.name}}{john@vdwetering.name}{https://orcid.org/0000-0002-5405-8959}{NWO Rubicon Personal Grant}
\author{Lia Yeh}{Department of Computer Science, University of Oxford, United Kingdom}{lia.yeh@cs.ox.ac.uk}{https://orcid.org/0000-0003-2704-4057}{Oxford -
Basil Reeve Graduate Scholarship at Oriel College with the Clarendon Fund}
\authorrunning{A.~Glaudell, N.J.~Ross, J.~van de Wetering, L.~Yeh}

\Copyright{A.~Glaudell, N.J.~Ross, J.~van de Wetering, L.~Yeh}
\nolinenumbers
\hideLIPIcs

\begin{document}
\maketitle
\begin{abstract}
    A popular universal gate set for quantum computing with qubits is Clifford+$T$, as this can be readily implemented on many fault-tolerant architectures. For qutrits, there is an equivalent $T$ gate, that, like its qubit analogue, makes Clifford+$T$ approximately universal, is injectable by a magic state, and supports magic state distillation. 
    However, it was claimed that a better gate set for qutrits might be Clifford+$R$, where $R=\text{diag}(1,1,-1)$ is the \emph{metaplectic} gate, as certain protocols and gates could more easily be implemented using the $R$ gate than the $T$ gate. 
    In this paper we show that when we have at least two qutrits, the qutrit Clifford+$R$ unitaries form a strict subset of the Clifford+$T$ unitaries, by finding a direct decomposition of $R \otimes \mathbb{I}$ as a Clifford+$T$ circuit and proving that the $T$ gate cannot be exactly synthesized in Clifford+$R$. This shows that in fact the $T$ gate is at least as powerful as the R gate, up to a constant factor. Moreover,
    we additionally show that it is impossible to find a single-qutrit Clifford+$T$ decomposition of the $R$ gate, making our result tight.
\end{abstract}

\section{Introduction}
Most theoretical work on quantum computing has focussed on qu\emph{bit}s, two-dimensional quantum systems. However, many proposed physical types of qubits are actually restricted subspaces of higher-dimensional systems, where the natural dimension can be much higher. The restriction to qubits is made for two reasons: the difficulty of precisely controlling quantum systems and the reliance on analogy to classical computers where two-valued bits reign supreme. However, as quantum control continues to improve, researchers have revisited this design choice. In some cases, higher-dimensional qu\emph{dit}s appear to be the superior option. 
For example, using the otherwise wasted subspace of a quantum system as a qudit drastically increases a device's information density compared to a qubit-based counterpart. Interest in qudit algorithms and physical implementations has risen recently, due to the potential advantages in runtime efficiency, resource requirements, computational space, and noise resilience in communication~\cite{WangY2020quditsreview}.

For qudits to make a good foundation for a quantum computer, we need methods to achieve fault-tolerance. In qubit-based protocols, one popular paradigm is to rely on the Clifford+$T$ gate set. This gate set consists of the efficiently simulable Clifford gates that can be implemented directly on many error correcting codes, and the $T$ gate that can be implemented by distilling and injecting magic states~\cite{BravyiS2005ftqc}. Analogous constructions have been developed for qudits of all dimensions, each one relying on a specific generalization of the Clifford+$T$ gate set~\cite{CampbellE2012tgatedistillation}.

In this paper we focus on the case of qu\emph{trit}s, three-dimensional quantum systems. Qutrits permit the qutrit Clifford+$T$ gate set, which can be implemented fault-tolerantly on qutrit error correcting codes, analogous to the qubit setting~\cite{CuiS2015universalmetaplectic}.
The Clifford+$T$ gate set is however not the only proposed universal fault-tolerant gate set for qutrits.
In a series of papers~\cite{CuiS2015universalmetaplectic, CuiS2015universalweakly, BocharovA2016topologicalcompilation, BocharovA2016ternaryarithmetics, BocharovA2016optimalitymetaplectic, BocharovA2017ternaryshor} and a patent~\cite{Microsoft2021metaplecticanyonpatent}, the non-Clifford qutrit gate of choice is the $R$ gate, also referred to as the FLIP gate, \emph{reflection} gate, or \emph{metaplectic} gate. It is defined as $R\coloneqq\text{diag}(1,1,-1)$.  This gate was defined in Ref.~\cite{AnwarH2012r2distillation}, where it was shown to admit a magic state distillation and injection protocol.  As a non-Clifford gate it achieves approximate universality when added to the Clifford gate set~\cite{GottesmanD1999ftqudit}, as explicitly proved in Ref.~\cite[Theorem 2]{CuiS2015universalmetaplectic}.
It can be implemented in a framework of certain weakly-integral non-abelian anyons via braiding and topological measurement~\cite{CuiS2015universalweakly,CuiS2015universalmetaplectic}.

While the definition of the $R$ gate looks very simple, containing only $1$'s, $0$'s and a~$-1$, it is in fact nowhere in the qutrit Clifford hierarchy~\cite{CuiS2017Diagonalhierarchy}. This is because for qutrits, Clifford gates are based on the third root of unity $\omega = e^{i2\pi/3}$.  Despite this fact, the $R$ gate can still be injected into a qutrit circuit using a repeat-until-success procedure~\cite{AnwarH2012r2distillation}.  
The $R$ state has a magic state distillation protocol, after which such a repeat-until-success injection protocol can be applied to realise the $R$ gate~\cite{AnwarH2012r2distillation} fault-tolerantly.
Another construction of $R$ is by a measurement-assisted repeat-until-success protocol requiring two ancillary qutrits to probabilistically realize it out of Clifford gates~\cite{CuiS2015universalweakly}.
The $R$ gate has been suggested to be ``more powerful in practice'' than the $T$ gate by Bocharov, Roetteler, and Svore~\cite{BocharovA2017ternaryshor}.  They computed the cost of approximating the third level of the Clifford hierarchy in the Clifford+$R$ (which they refer to as the \textit{metaplectic}) gate set, and claimed that constructing the $R$ gate in the Clifford+$T$ gate set requires multiple ancillae and repeat-until-success circuits.

In this paper we demonstrate that, contrary to these previous assertions, $T$ is strictly more powerful than $R$. We show that while no single-qutrit Clifford+$T$ circuit composes to an $R$ gate unitarily%
\footnote{Unless explicitly stated, we mean ``single-qutrit'' circuit to be ancilla-free.}%
, rather unexpectedly the $R$ gate is exactly constructible through a unitary two-qutrit Clifford+$T$ circuit with $T$-count 39, which we construct in Section~\ref{sec:constructrtwo}. This demonstrates that $R$ $\in$ Clifford+$T$. 
Additionally, we prove that the converse is not true, i.e.~that $T$ $\notin$ Clifford+$R$, and hence Clifford+$R$ $\subsetneq$ Clifford+$T$. This directly implies any Clifford+$T$ computation can be exactly implemented through Clifford+$T$ gates with constant overhead, whereas there exist Clifford+$T$ circuits whose implementation via Clifford+$R$ must strictly increase with the desired precision.

This result might seem to contradict the fact that $R$ does not belong anywhere in the Clifford hierarchy, while every Clifford gate and the $T$ gate belongs to the third level $C_3$. But recall that while $C_1$ and $C_2$ are closed under composition, this is no longer true for the higher levels of the Clifford hierarchy. In particular, it is not true that any circuit built out of Clifford+$T$ gates is a unitary that belongs to $C_3$.

The paper is structured as follows.
We cover all the basics on qutrit quantum computation and gate synthesis in Section~\ref{sec:preliminaries}.
Then in Section~\ref{sec:constructrtwo} we show how to build the $R$ gate as a two-qutrit unitary using only Clifford+$T$ gates and we prove that it is not possible to do this using just single-qutrit Clifford+$T$ gates. We finish by demonstrating that $T$ is \emph{not} an element of Clifford+$R$ so that Clifford+$R$ is in fact a strict subset of Clifford+$T$.
We end with some concluding remarks in Section~\ref{sec:conclusion}.

\section{Qutrit Clifford+\texorpdfstring{$T$}{T}}
\label{sec:preliminaries}

A qubit is a two-dimensional Hilbert space. Similarly, a qutrit is a three-dimensional Hilbert space. We will write $\ket{0}$, $\ket{1}$, and $\ket{2}$ for the standard computational basis states of a qutrit.
Any normalised qutrit state can then be written as
\begin{equation}
   \ket{\psi} = \alpha \ket{0} +  \beta \ket{1} + \gamma \ket{2} 
\end{equation}
where $\alpha,\beta,\gamma\in \mathbb{C}$ and $|\alpha|^2 + |\beta|^2 + |\gamma|^2 = 1$.

For a comprehensive overview of quantum computing based on qudits, we refer to the 2020 review by Wang, Hu, Sanders, and Kais~\cite{WangY2020quditsreview}.  A qudit quantum processor has been experimentally demonstrated on ion trap systems~\cite{RingbauerM2021quditions} and superconducting circuits~\cite{BlokM2021scrambling,YeB2018cphasephoton,YurtalanM2020Walsh-Hadamard}.

\subsection{Pauli gates and permutation gates}
\label{sec:Xgates}

Several concepts for qubits extend to qutrits, or more generally to qu\emph{dits}, which are $d$-dimensional quantum systems.  We are concerned with the qudit generalizations of Paulis and Cliffords.
\begin{definition}
For a $d$-dimensional qudit, the Pauli $X$ and $Z$ gates are defined as
\begin{equation}
    X\ket{k} = \ket{k+1} \qquad\qquad Z\ket{k} = \omega^k \ket{k}
\end{equation}
where $\omega:= e^{2\pi i/d}$ such that $\omega^d = 1$, and the addition $\ket{k+1}$ is taken modulo $d$.
We define the \emph{Pauli group} as the set of unitaries generated by compositions and tensor products of the $X$ and $Z$ gates. 
We write $\mathcal{P}_n^d$ for the Pauli group on $n$ qudits~\cite{GottesmanD1999ftqudit,HowardM2012quditTgate}.
\end{definition}
For qubits this $X$ gate is just the NOT gate while $Z=\text{diag}(1,-1)$. For the duration of this paper we will work solely with qutrits, so we take $\omega$ to always be equal to $e^{2\pi i/3}$.

For a qubit there is only one non-trivial permutation of the standard basis states, which is implemented by the $X$ gate.
For qutrits, this grows to five non-trivial permutations of the basis states. We call these $\tau$ gates and we specify them as $\tau_{L}$ where $L$ is a permutation of the elements $\{0,1,2\}$ written in cycle notation. For example, $\tau_{(02)}$ is the permutation which maps $\ket{0}\mapsto\ket{2}$, $\ket{1}\mapsto\ket{1}$, and $\ket{2}\mapsto\ket{0}$. The five non-trivial permutations are then $\tau_{(01)}$, $\tau_{(12)}$, $\tau_{(02)}$, $\tau_{(012)}$, and $\tau_{(021)}$ along with the trivial identity permutation $\mathbb{I}=\tau_{(0)(1)(2)}$. Compositions of these operators are given by $\tau_L\cdot \tau_M=\tau_{L\cdot M}$ with $L\cdot M$ the composition of permutations. Note that $\tau_{(012)}=X$ and $\tau_{(021)}=X^\dagger$.

\subsection{Exact synthesis and number rings}

One natural question to ask when given a set of gates is to determine which operations can be implemented as a circuit over those gates. This is called the \emph{exact synthesis} problem. One frequently useful notion in addressing exact synthesis is computing the matrix representations of the set of gates in the computational basis and characterizing the \emph{number ring} to which their entries belong. A number ring is a set of numbers which explicitly contains 0 and 1 and is closed under the operations of addition and multiplication. For example, the integers $\mathbb{Z}$ form a number ring. 

We can \emph{extend} number rings by considering what happens when we introduce new numbers to the number ring. When we extend the number ring $R$ by $\alpha$ we write $R[\alpha]$ for the ring of formal sums $\sum_j r_j \alpha^j$ where $r_j\in R$.
Generally, we extend by an $\alpha$ which is the root of some monic polynomial whose coefficients come from $R$. If that polynomial has degree $p$, then all powers of $\alpha$ which are greater than $p-1$ and appear in an element of $R[\alpha]$ can be reduced via that polynomial.
For example, the third root of unity $\omega$ solves the monic polynomial $1+\omega+\omega^2=0$
	over the integers so that we define
	\(
		\mathbb{Z}[\omega] = \{a+b\omega~|~a,b\in\mathbb{Z}\}.
	\)
	Any higher-order powers of $\omega$ which might appear in an element of $\mathbb{Z}[\omega]$ can be reduced through its polynomial as for example, $\omega^2=-1-\omega$ and $\omega^3=1$.

Another common way to modify number rings is to introduce new denominators by \emph{localizing} a ring. For a number ring $R$ we can take any multiplicatively-closed subset $\mu$ of $R$ which contains 1 but not 0 and introduce that set of numbers as denominators:
\[
	\mu^{-1}R = \left\{\frac{r}{m}~\middle|~r\in R~\text{and}~m\in \mu\right\}.
\]

\begin{definition}
	The ring of \emph{triadic fractions} is the number ring defined by localizing $\mathbb{Z}$ at the set $\mu=\{3^k~|~k\in\mathbb{N}\}$, which we denote as $\mathbb{T} := \mu^{-1}\mathbb{Z} = \{a/3^k~\mid~a\in \mathbb{Z}, k\in \mathbb{N}\}$.
\end{definition}

The use of number rings to help solve the exact synthesis problem
stems from the following statement, attributable to many authors in
the field but perhaps most notably Kliuchnikov, Maslov, and Mosca
\cite{KliuchnikovV2013qubitctrlt}:

\begin{lemma}
	\label{lemma:gatesetring}
	Let $\mathcal{G}=\{G_1,\cdots,G_k\}$ be a quantum gate set. For all $j\in\{1,\dots,k\}$, let each $G_j$ have the computational basis matrix representation $M_j$ up to a complex global phase such that $M_j$ is a matrix with entries in the number ring $R$. Up to a global phase, the matrix representation for any circuit over $\mathcal{G}$ only has entries in the number ring $R$.
\end{lemma}

It is important to note that Lemma~\ref{lemma:gatesetring} only suffices to \emph{exclude} operations from being representable over a given gate set. To show that a circuit with entries in a particular number ring implies expressibility over a certain gate set is generally equivalent to providing a full solution to the exact synthesis problem.

\begin{example}
	Any qutrit Pauli operation in the computational basis has entries from the number ring $\mathbb{Z}[\omega]$. This follows directly from Lemma~\ref{lemma:gatesetring} and the fact that $\mathcal{P}_n^3$ is generated by~$X$ and $Z$.
\end{example}

One interesting aspect of $\mathbb{Z}[\omega]$ (and number rings which contain roots of unity in general) is that it contains elements which square to non-square integers. 
In particular, $(\omega-\omega^2)^2=(\omega^2-\omega)^2=-3$. 
Note the minus sign here, which is important as $\pm\sqrt{3}\not\in\mathbb{Z}[\omega]$. Due to the ubiquity of the Pauli group and the natural appearance of $\omega$, when working with circuits over qutrits it has become increasingly customary to use $\pm(\omega-\omega^2) = \pm i\sqrt{3}$ in place of $\sqrt{3}$ when possible. We make use of this replacement frequently (see, e.g., the Hadamard gate defined below).

\subsection{Clifford gates}

Another concept that translates to qutrits (or more general qudits) is that of Clifford unitaries.

\begin{definition}
    Let $U$ be a unitary acting on $n$ qudits. We say that $U$ is \emph{Clifford} when every Pauli is mapped to another Pauli under conjugation by $U$. I.e., for any $P\in \mathcal{P}_n^d$ we have $UPU^\dagger \in \mathcal{P}_n^d$.
\end{definition}

Note that the set of $n$-qudit Cliffords forms a group under composition. For qubits, this group is generated by the $S$, Hadamard, and CX gates. The same is true for qutrits, for the right generalisation of these gates.

\begin{definition}\label{def:Z-phase-gate}
    We write $Z(a,b)$ for the \emph{phase gate} that acts as $Z(a,b)\ket{0} = \ket{0}$, $Z(a,b)\ket{1} = \omega^a\ket{1}$ and $Z(a,b)\ket{2} = \omega^b\ket{2}$ where we take $a,b\in \mathbb{R}$.
\end{definition}
We define $Z(a,b)$ in this way, taking $a$ and $b$ to correspond to phases that are multiples of $\omega$, because $Z(a,b)$ will turn out to be Clifford iff $a$ and $b$ are integers. Note that the collection of all $Z(a,b)$ operators constitutes the group of diagonal single-qutrit unitaries modded out by a global phase. Composition of these operations is given by $Z(a,b)\cdot Z(c,d)=Z(a+b,c+d)$.

\begin{definition}
    The qutrit $S$ gate is $S:= Z(0,1)$. I.e., it multiplies the $\ket{2}$ state by $\omega$.
\end{definition}

For qubits, the Hadamard interchanges the $Z$ basis, $\ket{0}$ and $\ket{1}$, which are the eigenstates of the Pauli $Z$, and the $X$ basis, consisting of the states $\ket{\pm} := \frac{1}{\sqrt{2}}(\ket{0}\pm \ket{1})$. The same holds for the qutrit Hadamard.
In this case the $X$ basis consists of the following states (where we recall from above that $\frac{1}{\omega^2-\omega} = i/\sqrt{3}$):
\begin{align}
    \ket{+} &:= \frac{1}{\omega^2-\omega} (\ket{0}+\ket{1}+\ket{2}) \\
    \ket{\omega} &:= \frac{1}{\omega^2-\omega} (\ket{0}+\omega\ket{1}+\omega^2\ket{2}) \\
    \ket{\omega^2} &:= \frac{1}{\omega^2-\omega} (\ket{0}+\omega^2\ket{1}+\omega\ket{2})
\end{align}

\begin{definition}
    The \emph{qutrit Hadamard gate} $H$ is the gate that maps $\ket{0} \mapsto \ket{+}$, $\ket{1}\mapsto \ket{\omega}$ and $\ket{2} \mapsto \ket{\omega^2}$. As a matrix:
    \begin{equation}
	    H \ := \ \frac{1}{\omega^2-\omega}\begin{pmatrix}
    	1 & 1 & 1 \\
    	1 & \omega & \omega^2 \\
    	1 & \omega^2 & \omega
    \end{pmatrix}
    \end{equation}
\end{definition}

Note that, unlike the qubit Hadamard, the qutrit Hadamard is \emph{not} self-inverse. 
In fact, we have $H^2=-\tau_{(12)}$, so that $H^4 = \mathbb{I}$. In particular, $H^\dagger = H^3$. Furthermore, we note that just as the Clifford group in qubits generates certain global phases, the relation
\(
	(SH)^3=-\omega
\)
implies that global phases of $\pm1$, $\pm\omega$, and $\pm\omega^2$ naturally appear in the qutrit Clifford group.
The Pauli and $S$ gates we defined all have matrix representations with entries over $\mathbb{Z}[\omega]$. We see that $H$ naturally introduces denominators into our matrices, and so we should localize $\mathbb{Z}[\omega]$ to ensure we can characterize circuits which contain $H$. Since
\[
	\frac{\omega^k}{\omega^2-\omega}=\frac{\omega^k(\omega-\omega^2)}{3}
\]
we can introduce the appropriate denominators by localizing at $\mu=\{3^k~|~k\in\mathbb{N}\}$ to get the number ring $\mu^{-1}\mathbb{Z}[\omega]$. Note that this is equivalent to the number ring $\mathbb{T}[\omega]$ which consists of elements $a+b\omega$ where $a,b\in \mathbb{T}$ are triadic fractions.

In Definition~\ref{def:Z-phase-gate} we defined the Z phase gate. Similarly, we can define the X phase gates, that give a phase to the X basis states.
\begin{definition}\label{def:X-phase-gate}
    We define the \emph{X phase gates} to be $X(a,b) := HZ(a,b)H^\dagger$ where $a,b \in \mathbb{R}$.
\end{definition}
We already saw examples of such X phase gates: $X = X(2,1)$ and $X^\dagger = X(1,2)$.

Any single-qutrit Clifford can be represented (up to global phase) as a composition of Clifford $Z$ and $X$ phase gates.
In particular, we can represent the qutrit Hadamard as follows~\cite{GongX2017equivalence}:
\begin{align}
\label{eq:EUrule}
H &= -Z(2,2)X(2,2)Z(2,2) = -X(2,2)Z(2,2)X(2,2) \\
H^\dagger &= -Z(1,1)X(1,1)Z(1,1) = -X(1,1)Z(1,1)X(1,1) \label{eq:EUrule2}
\end{align}

The final Clifford gate we need is the qutrit CX.

\begin{definition}
    The qutrit CX gate is the two-qutrit gate defined by $\text{CX}\ket{i,j} = \ket{i,i+j}$ where the addition is taken modulo $3$.
\end{definition}

\begin{proposition}
	\label{prop:clifford}
    Let $U$ be a qutrit Clifford unitary. Then up to global phase $U$ can be written as a composition of the $S$, $H$ and CX gates~\cite{GottesmanD1999ftqudit}.
\end{proposition}
From this it easily follows that the $Z(a,b)$ and $X(a,b)$ gates are Clifford if and only if~$a$ and~$b$ are integers.

\begin{corollary}
	Let $U$ be a qutrit Clifford unitary. Then up to a global phase $U$ has a matrix representation in the computational basis with entries in the number ring $\mathbb{T}[\omega]$. 
\end{corollary}
\begin{proof}
	This follows from Proposition~\ref{prop:clifford}, the definitions of $S$, $H$, and CX, and Lemma~\ref{lemma:gatesetring}.
\end{proof}

\subsection{T gates and qutrit controlled gates}

Clifford unitaries don't suffice for universal computation, so let's introduce the $T$ gate. For this we will need the ninth-root of unity.
Throughout the remainder of the paper, we define $\zeta = e^{2\pi i/9}$.
\begin{definition}
    The qutrit $T$ gate is the Z phase gate defined as $T:=Z(1/3,-1/3) = \text{diag}(1,\zeta,\zeta^8)$~\cite{PrakashS2018normalform,CampbellE2012tgatedistillation,HowardM2012quditTgate}.
\end{definition}
Like the qubit $T$ gate, the qutrit $T$ gate belongs to the third level of the Clifford hierarchy, can be injected into a circuit using magic states, and its magic states can be distilled by magic state distillation. This means that we can fault-tolerantly implement this qutrit $T$ gate on many types of quantum error correcting codes.
Also as for qubits, the qutrit Clifford+$T$ gate set is approximately universal, meaning that we can approximate any qutrit unitary using just Clifford gates and the $T$ gate~\cite[Theorem 1]{CuiS2015universalmetaplectic}. 

The $T$ gate introduces the phase $\zeta$ into matrix representations of circuits and thus we should consider extending the previously-defined $\mathbb{T}[\omega]$ by $\zeta$. Note that $\zeta$ is a ninth root of unity which solves the cubic polynomial
\begin{equation}
	\zeta^3-\omega=0\label{eq:relation2}
\end{equation}
over $\mathbb{T}[\omega]$. In fact, this polynomial has no solutions over $\mathbb{T}[\omega]$, implying that $\zeta\not\in\mathbb{T}[\omega]$ (see Appendix~\ref{app:rootofunity}). We thus define the number ring $\mathbb{T}[\zeta]$:
\begin{definition}
	The extension of $\mathbb{T}[\omega]$ by $\zeta$ is the number ring $\mathbb{T}[\omega][\zeta]\cong\mathbb{T}[\zeta]$ defined by
	\begin{equation*}
	\mathbb{T}[\omega][\zeta]\cong\mathbb{T}[\zeta] := \{a+b\zeta+c\zeta^2+d\zeta^3+e\zeta^4+f\zeta^5~|~a,b,c,d,e,f\in\mathbb{T}\}.
	\end{equation*}
\end{definition}
Any higher powers of $\zeta$ that might appear in an expression for an element of $\mathbb{T}[\zeta]$ can be reduced using for instance Eq.~\eqref{eq:relation2}.

\begin{lemma}
	Let $U$ be a qutrit Clifford+$T$ unitary. Then up to a global phase $U$ has a matrix representation in the computational basis with entries in the number ring $\mathbb{T}[\zeta]$.
\end{lemma}
\begin{proof}
	By the definitions of $S$, $H$, $T$, and CX  and Lemma~\ref{lemma:gatesetring}.
\end{proof}

Using $T$ gates, we can construct certain controlled unitaries.
When we have an $n$-qu\emph{bit} unitary $U$, we can speak of the controlled gate that implements $U$. This is the $(n+1)$-qubit gate that acts as the identity when the first qubit is in the $\ket{0}$ state, and implements $U$ on the last $n$ qubits if the first qubit is in the $\ket{1}$ state.

For qutrits there are however multiple notions of control.

\begin{definition}
Let $U$ be a qutrit unitary. Then the \emph{$\ket{2}$-controlled $U$} is the unitary $\ket{2}$-$U$ that acts as
\begin{equation*}
    \ket{0}\otimes \ket{\psi} \mapsto \ket{0}\otimes \ket{\psi} \qquad
    \ket{1}\otimes \ket{\psi} \mapsto \ket{1}\otimes \ket{\psi} \qquad
    \ket{2}\otimes \ket{\psi} \mapsto \ket{2}\otimes U\ket{\psi}
\end{equation*}
I.e., it implements $U$ on the last qutrits if and only if the first qutrit is in the $\ket{2}$ state.
\end{definition}
Note that by conjugating the first qutrit with $X$ or $X^\dagger$ gates we can make the gate also be controlled on the $\ket{1}$ or $\ket{0}$ state.

A different notion of qutrit control was introduced by Bocharov, Roetteler, and Svore~\cite{BocharovA2017ternaryshor}:
Given a qutrit unitary $U$ they define $\Lambda(U)\ket{c}\ket{t} = \ket{c} \otimes (U^c \ket{t})$.
I.e., apply the unitary $U$ a number of times equal to to the value of the control qutrit, so that if the control qutrit is $\ket{2}$ we apply $U^2$ to the target qutrits. Note that we can get this notion of control from the former one: just apply a $\ket{1}$-controlled $U$, followed by a $\ket{2}$-controlled $U^2$. 
The Clifford CX gate defined earlier is in this notation equal to $\Lambda(X)$.

Adding controls to a Clifford gate generally makes it non-Clifford. In the case of the CX gate, which is $\Lambda(X)$, it is still Clifford, but the $\ket{2}$-controlled $X$ is \emph{not}.
\begin{lemma}
The $\ket{2}$-controlled $X$ and $\ket{2}$-controlled $X^\dagger$ gates can be implemented unitarily using Clifford+$T$ gates without ancillae, with a $T$-count of 3.
\end{lemma}
\begin{proof}
    The $\ket{2}$-controlled $X$ gate can be constructed by the following 3 $T$ gate circuit:
    
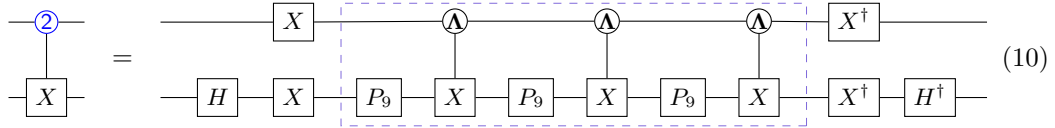
\begin{equation}
        \label{eq:tcxp1}
        \begin{tikzpicture}
	\begin{pgfonlayer}{nodelayer}
		\node [style=2 control] (0) at (-11.75, 1) {\textcolor{blue}{\textsf{2}}};
		\node [style=box] (1) at (-11.75, -1) {$X$};
		\node [style=none] (3) at (-9.75, 0) {=};
		\node [style=box] (4) at (-7.25, -1) {$H$};
		\node [style=box] (5) at (-5.25, -1) {$X$};
		\node [style=box] (6) at (-3, -1) {$P_{9}$};
		\node [style=box] (7) at (-1, -1) {$X$};
		\node [style=box] (8) at (1, -1) {$P_{9}$};
		\node [style=box] (9) at (3, -1) {$X$};
		\node [style=box] (10) at (5, -1) {$P_{9}$};
		\node [style=box] (11) at (7, -1) {$X$};
		\node [style=box] (12) at (9.5, -1) {$X^\dagger$};
		\node [style=box] (13) at (11.5, -1) {$H^\dag$};
		\node [style=none] (14) at (13, -1) {};
		\node [style=vertex set] (15) at (7, 1.025) {$\Lambda$};
		\node [style=vertex set] (16) at (-1, 1.025) {$\Lambda$};
		\node [style=vertex set] (17) at (3, 1.025) {$\Lambda$};
		\node [style=box] (18) at (9.5, 1) {$X^\dagger$};
		\node [style=none] (19) at (-8.75, 1) {};
		\node [style=none] (20) at (-8.75, -1) {};
		\node [style=none] (21) at (13, 1) {};
		\node [style=none] (22) at (-4, 1.5) {};
		\node [style=none] (23) at (-4, -1.75) {};
		\node [style=none] (24) at (8.25, -1.75) {};
		\node [style=none] (25) at (8.25, 1.5) {};
		\node [style=none] (26) at (-12.75, 1) {};
		\node [style=none] (27) at (-12.75, -1) {};
		\node [style=none] (28) at (-10.75, 1) {};
		\node [style=none] (29) at (-10.75, -1) {};
		\node [style=box] (30) at (-5.25, 1) {$X$};
	\end{pgfonlayer}
	\begin{pgfonlayer}{edgelayer}
		\draw [in=90, out=-90] (0) to (1);
		\draw (4) to (5);
		\draw (5) to (6);
		\draw (6) to (7);
		\draw (7) to (8);
		\draw (8) to (9);
		\draw (9) to (10);
		\draw (10) to (11);
		\draw (11) to (12);
		\draw (12) to (13);
		\draw (16) to (7);
		\draw (17) to (9);
		\draw (15) to (11);
		\draw (15) to (18);
		\draw (16) to (17);
		\draw (17) to (15);
		\draw (18) to (21.center);
		\draw (20.center) to (4);
		\draw [style=dashed thin purple] (24.center)
			 to (25.center)
			 to (22.center)
			 to (23.center)
			 to cycle;
		\draw (26.center) to (0);
		\draw (0) to (28.center);
		\draw (27.center) to (1);
		\draw (1) to (29.center);
		\draw [style=simple] (13) to (14.center);
		\draw [style=simple] (19.center) to (30);
		\draw [style=simple] (30) to (16);
	\end{pgfonlayer}
\end{tikzpicture}
    \end{equation}
    Here $P_9 := X T X^\dagger$.
    The dashed box shows the construction by Bocharov, Roetteler, and Svore~\cite[Figure 6]{BocharovA2017ternaryshor} for a Clifford equivalent gate, that our construction is based on.
    Taking the adjoint of Eq.~\eqref{eq:tcxp1} gives the $\ket{2}$-controlled $X^\dagger$ gate.
\end{proof}

\begin{lemma}
    \label{lemma:tcd12}
    The $\ket{2}$-controlled versions of the $\tau_{(01)}$, $\tau_{(02)}$, and $\tau_{(12)}$ gates can be implemented unitarily using Clifford+$T$ gates without ancillae, with a $T$-count of 15.
\end{lemma}
\begin{proof}
    The $\ket{2}$-controlled $\tau_{(12)}$ gate can be constructed as follows:
    
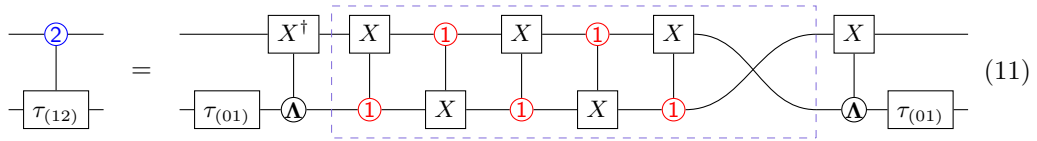
\begin{equation}
        \label{eq:tcd12}
        \begin{tikzpicture}
	\begin{pgfonlayer}{nodelayer}
		\node [style=2 control] (0) at (-13.25, 1) {\textcolor{blue}{\textsf{2}}};
		\node [style=box] (1) at (-13.25, -1) {$\tau_{(12)}$};
		\node [style=none] (3) at (-11, 0) {=};
		\node [style=box] (4) at (-7, 1) {$X^\dagger$};
		\node [style=box] (6) at (-5, 1) {$X$};
		\node [style=1 control] (7) at (-3, 1) {\textcolor{red}{\textsf{1}}};
		\node [style=box] (8) at (-1, 1) {$X$};
		\node [style=1 control] (9) at (1, 1) {\textcolor{red}{\textsf{1}}};
		\node [style=box] (10) at (3, 1) {$X$};
		\node [style=none] (11) at (6.75, 1) {};
		\node [style=box] (12) at (9.5, -1) {$\tau_{(01)}$};
		\node [style=1 control] (13) at (-5, -1) {\textcolor{red}{\textsf{1}}};
		\node [style=1 control] (14) at (-1, -1) {\textcolor{red}{\textsf{1}}};
		\node [style=1 control] (15) at (3, -1) {\textcolor{red}{\textsf{1}}};
		\node [style=none] (16) at (-10, 1) {};
		\node [style=none] (17) at (-10, -1) {};
		\node [style=none] (18) at (10.75, 1) {};
		\node [style=none] (19) at (10.75, -1) {};
		\node [style=none] (20) at (-14.5, 1) {};
		\node [style=none] (21) at (-14.5, -1) {};
		\node [style=none] (22) at (-12, 1) {};
		\node [style=none] (23) at (-12, -1) {};
		\node [style=none] (24) at (6.75, -1) {};
		\node [style=box] (25) at (-3, -1) {$X$};
		\node [style=box] (26) at (1, -1) {$X$};
		\node [style=vertex set] (27) at (-7, -1) {$\Lambda$};
		\node [style=box] (28) at (-8.75, -1) {$\tau_{(01)}$};
		\node [style=none] (29) at (-6, 1.75) {};
		\node [style=none] (30) at (-6, -1.75) {};
		\node [style=none] (31) at (6.75, -1.75) {};
		\node [style=none] (32) at (6.75, 1.75) {};
		\node [style=box] (33) at (7.75, 1) {$X$};
		\node [style=vertex set] (34) at (7.75, -1) {$\Lambda$};
	\end{pgfonlayer}
	\begin{pgfonlayer}{edgelayer}
		\draw (0) to (1);
		\draw (6) to (7);
		\draw (7) to (8);
		\draw [in=180, out=0] (8) to (9);
		\draw (9) to (10);
		\draw (8) to (14);
		\draw (6) to (13);
		\draw (10) to (15);
		\draw (12) to (19.center);
		\draw (20.center) to (0);
		\draw (0) to (22.center);
		\draw (21.center) to (1);
		\draw (1) to (23.center);
		\draw [in=0, out=180] (24.center) to (10);
		\draw [in=-180, out=0] (15) to (11.center);
		\draw (13) to (25);
		\draw (25) to (14);
		\draw (14) to (26);
		\draw (26) to (15);
		\draw (9) to (26);
		\draw (25) to (7);
		\draw (27) to (13);
		\draw (17.center) to (28);
		\draw (28) to (27);
		\draw (4) to (6);
		\draw (16.center) to (4);
		\draw (4) to (27);
		\draw [style=dashed thin purple] (31.center)
			 to (32.center)
			 to (29.center)
			 to (30.center)
			 to cycle;
		\draw [style=simple] (24.center) to (34);
		\draw [style=simple] (34) to (12);
		\draw [style=simple] (11.center) to (33);
		\draw [style=simple] (33) to (18.center);
		\draw [style=simple] (33) to (34);
	\end{pgfonlayer}
\end{tikzpicture}
    \end{equation}
    The dashed box shows the Clifford equivalent gate given by Bocharov, Roetteler, and Svore~\cite{BocharovA2017ternaryshor,BocharovA2016ternaryarithmetics} upon which the construction is based.
    The $\ket{2}$-controlled $\tau_{(01)}$ and $\ket{2}$-controlled $\tau_{(02)}$ gates follow from Clifford equivalence.
\end{proof}

\section{Results}
\label{sec:constructrtwo}
Previous implementations of the $R$ gate require either distillation~\cite{AnwarH2012r2distillation} or probabilistic creation of the diag($1,1,-1$) state~\cite{CuiS2015universalweakly,BocharovA2017ternaryshor}; both approaches then necessitate injection by a repeat-until-success protocol.  
Here we present a new approach, which implements $R$ unitarily over the qutrit Clifford+$T$ gate set. 
As we will discuss later, it is actually impossible to exactly build the $R$ gate from only single-qutrit Clifford+$T$ gates.  However, we \emph{can} construct the two-qutrit $R \otimes \mathbb{I}$ unitarily using Clifford+$T$ gates.
We will do this\footnote{Implemented at \url{https://github.com/lia-approves/qudit-circuits/tree/main/qutrit_R_from_T}.} by showing how to construct certain $\ket{2}$-controlled gates and then using the following observation.

Our construction will be based on the following idea:
\begin{equation}\label{eq:R-circuit}
        \begin{tikzpicture}
	\begin{pgfonlayer}{nodelayer}
		\node [style=2 control] (12) at (3.5, 0.75) {\textcolor{blue}{\textsf{2}}};
		\node [style=box] (13) at (3.5, -0.75) {$-\mathbb{I}$};
		\node [style=none] (14) at (2, -0.75) {};
		\node [style=none] (15) at (2, 0.75) {};
		\node [style=none] (16) at (5, 0.75) {};
		\node [style=none] (17) at (5, -0.75) {};
		\node [style=none] (18) at (0.5, 0) {$=$};
		\node [style=none] (25) at (6.5, 0) {$=$};
		\node [style=box] (26) at (9.5, 0.75) {$R$};
		\node [style=none] (27) at (8, -0.75) {};
		\node [style=none] (28) at (8, 0.75) {};
		\node [style=none] (29) at (11, 0.75) {};
		\node [style=none] (30) at (11, -0.75) {};
		\node [style=2 control] (31) at (-7.5, 0.75) {\textcolor{blue}{\textsf{2}}};
		\node [style=box] (32) at (-7.5, -0.75) {$-H^\dagger$};
		\node [style=none] (33) at (-9, -0.75) {};
		\node [style=none] (34) at (-9, 0.75) {};
		\node [style=2 control] (35) at (-2.5, 0.75) {\textcolor{blue}{\textsf{2}}};
		\node [style=box] (36) at (-2.5, -0.75) {$-H^2$};
		\node [style=none] (37) at (-1, 0.75) {};
		\node [style=none] (38) at (-1, -0.75) {};
		\node [style=2 control] (40) at (-5, 0.75) {\textcolor{blue}{\textsf{2}}};
		\node [style=box] (41) at (-5, -0.75) {$-H^\dagger$};
	\end{pgfonlayer}
	\begin{pgfonlayer}{edgelayer}
		\draw (12) to (13);
		\draw (14.center) to (13);
		\draw [style=simple] (15.center) to (12);
		\draw (13) to (17.center);
		\draw (16.center) to (12);
		\draw [style=simple] (28.center) to (26);
		\draw (29.center) to (26);
		\draw (27.center) to (30.center);
		\draw (31) to (32);
		\draw (33.center) to (32);
		\draw [style=simple] (34.center) to (31);
		\draw (35) to (36);
		\draw (36) to (38.center);
		\draw (37.center) to (35);
		\draw (40) to (41);
		\draw (31) to (40);
		\draw (40) to (35);
		\draw (32) to (41);
		\draw (41) to (36);
	\end{pgfonlayer}
\end{tikzpicture}
    \end{equation}

Note that global phases in gates are usually not relevant, but that they become relevant when adding control wires to the gate. Here we have a \emph{controlled} global phase because the global phase of $-1$ is applied to the target if and only if the control is in the $\ket{2}$ state.  Therefore, this is an instance of phase kickback: The action of the $\ket{2}$-controlled $-\mathbb{I}$ gate is identical to applying $R \otimes \mathbb{I}$, i.e. the $R$ gate to the control qutrit and identity on the target.

The $\ket{2}$-controlled $-H^2 = \tau_{(12)}$ was constructed as a Clifford+$T$ circuit decomposition in Eq.~\eqref{eq:tcd12}. It hence remains to show how we can construct the $\ket{2}$-controlled $-H^\dagger$ gate in Clifford+$T$.
We can build this using the $\ket{2}$-controlled $S^\dagger$ gate.
\begin{lemma}
    \label{lemma:tcspdagphase}
    The $\ket{2}$-controlled $S^\dagger$ gate can be constructed unitarily without ancillae, up to a controlled global phase of $\zeta$, using only Clifford+$T$ gates, with $T$-count~8.
\end{lemma}
\begin{proof}
    The correctness can be verified by direct computation of the following circuit.
    \begin{equation}
        \begin{tikzpicture}
	\begin{pgfonlayer}{nodelayer}
		\node [style=none] (1) at (-5.5, -0.625) {};
		\node [style=none] (2) at (-5.5, 0.875) {};
		\node [style=none] (4) at (-1.5, -0.625) {};
		\node [style=none] (5) at (-1.5, 0.875) {};
		\node [style=2 control] (6) at (-3.5, 0.875) {\textcolor{blue}{\textsf{2}}};
		\node [style=box] (7) at (-3.5, -0.625) {$\zeta S^\dagger$};
		\node [style=none] (13) at (0, 0) {$=$};
		\node [style=box] (47) at (17.25, -0.625) {$\tau_{(01)}$};
		\node [style=none] (48) at (18.75, 0.875) {};
		\node [style=2 control] (49) at (11.25, 0.875) {\textcolor{blue}{\textsf{2}}};
		\node [style=box] (50) at (11.25, -0.625) {$X^\dagger$};
		\node [style=none] (65) at (18.75, -0.625) {};
		\node [style=box] (66) at (15.25, -0.625) {$T$};
		\node [style=box] (67) at (13.25, -0.625) {$\tau_{(01)}$};
		\node [style=box] (68) at (9, -0.625) {$\tau_{(01)}$};
		\node [style=2 control] (69) at (3, 0.875) {\textcolor{blue}{\textsf{2}}};
		\node [style=box] (70) at (3, -0.625) {$X$};
		\node [style=box] (75) at (7, -0.625) {$T^\dagger$};
		\node [style=box] (76) at (5, -0.625) {$\tau_{(01)}$};
		\node [style=none] (79) at (1.5, -0.625) {};
		\node [style=none] (199) at (1.5, 0.875) {};
	\end{pgfonlayer}
	\begin{pgfonlayer}{edgelayer}
		\draw (6) to (7);
		\draw (2.center) to (6);
		\draw (6) to (5.center);
		\draw (4.center) to (7);
		\draw (7) to (1.center);
		\draw (49) to (50);
		\draw (48.center) to (49);
		\draw (65.center) to (47);
		\draw (69) to (70);
		\draw (50) to (68);
		\draw (47) to (66);
		\draw (66) to (67);
		\draw (67) to (50);
		\draw (68) to (75);
		\draw (75) to (76);
		\draw (76) to (70);
		\draw (49) to (69);
		\draw [style=simple] (69) to (199.center);
		\draw [style=simple] (70) to (79.center);
	\end{pgfonlayer}
\end{tikzpicture}
    \end{equation}
Alternatively, it is easy to see that this circuit does nothing if the first qutrit is in the $\ket{0}$ or $\ket{1}$ state, as the $\tau_{01}$ gates cancel, so that the $T$ can cancel with the $T^\dagger$.
Otherwise, if the first qutrit is $\ket{2}$, then the middle permutations combine to $\tau_{(01)} X^\dagger \tau_{(01)} = \tau_{(012)} = X$. When the $T^\dagger$ is pushed through this, the phases get permuted, and when combined with the $T$ gives an $S$ gate up to global phase.
\end{proof}
\begin{corollary}
    The $\ket{2}$-controlled $Z(1,1) = \text{diag}(1,\omega,\omega)$ gate can be constructed unitarily without ancillae, up to a controlled global phase of $\zeta^{-2} = \zeta^7$, using only Clifford+$T$ gates, with $T$-count~8.
\end{corollary}
\begin{proof}
    Use the following circuit:
    \begin{equation}
        \label{eq:tczwwphase}
        \begin{tikzpicture}
	\begin{pgfonlayer}{nodelayer}
		\node [style=none] (1) at (-5.5, -0.625) {};
		\node [style=none] (2) at (-5.5, 0.875) {};
		\node [style=none] (4) at (-1.5, -0.625) {};
		\node [style=none] (5) at (-1.5, 0.875) {};
		\node [style=2 control] (6) at (-3.5, 0.875) {\textcolor{blue}{\textsf{2}}};
		\node [style=box] (7) at (-3.5, -0.625) {$\zeta^7 Z(1,1)$};
		\node [style=none] (13) at (0, 0) {$=$};
		\node [style=box] (47) at (7.5, -0.625) {$\tau_{(02)}$};
		\node [style=none] (48) at (9.5, 0.875) {};
		\node [style=2 control] (49) at (5.5, 0.875) {\textcolor{blue}{\textsf{2}}};
		\node [style=box] (50) at (5.5, -0.625) {$\zeta S^\dagger$};
		\node [style=none] (65) at (9.5, -0.625) {};
		\node [style=box] (76) at (3.5, -0.625) {$\tau_{(02)}$};
		\node [style=none] (79) at (1.5, -0.625) {};
		\node [style=none] (199) at (1.5, 0.875) {};
	\end{pgfonlayer}
	\begin{pgfonlayer}{edgelayer}
		\draw (6) to (7);
		\draw (2.center) to (6);
		\draw (6) to (5.center);
		\draw (4.center) to (7);
		\draw (7) to (1.center);
		\draw (49) to (50);
		\draw (48.center) to (49);
		\draw (65.center) to (47);
		\draw [style=simple] (199.center) to (49);
		\draw [style=simple] (47) to (50);
		\draw [style=simple] (50) to (76);
		\draw [style=simple] (76) to (79.center);
	\end{pgfonlayer}
\end{tikzpicture}
    \end{equation}
    Its correctness can be verified by direct computation, or by commuting $S^\dagger$ and $\tau_{(02)}$.
\end{proof}

\begin{lemma}
    \label{lemma:tcmhdag}
    The $\ket{2}$-controlled $-H^\dagger$ gate can be constructed unitarily without ancillae using Clifford+$T$ gates with $T$-count~24.
\end{lemma}
\begin{proof}
    In the construction given below, we use the decomposition of $-H^\dagger$ into alternating Z and X Clifford rotations of Eq.~\eqref{eq:EUrule2}.
    \begin{equation}
        \label{eq:tcmhdag}
        \begin{tikzpicture}
	\begin{pgfonlayer}{nodelayer}
		\node [style=2 control] (76) at (-6, 0.75) {\textcolor{blue}{\textsf{2}}};
		\node [style=box] (77) at (-6, -0.75) {$\zeta^{7} X(1,1)$};
		\node [style=box] (78) at (-9.25, -0.75) {$\zeta^{7} Z(1,1)$};
		\node [style=2 control] (79) at (-9.25, 0.75) {\textcolor{blue}{\textsf{2}}};
		\node [style=box] (80) at (-2.75, -0.75) {$\zeta^{7} Z(1,1)$};
		\node [style=2 control] (81) at (-2.75, 0.75) {\textcolor{blue}{\textsf{2}}};
		\node [style=none] (82) at (0, 0) {$=$};
		\node [style=box] (83) at (-10.5, 0.75) {$S^\dagger$};
		\node [style=box] (84) at (12, -0.75) {$-H^\dagger$};
		\node [style=2 control] (85) at (12, 0.75) {\textcolor{blue}{\textsf{2}}};
		\node [style=2 control] (86) at (5, 0.75) {\textcolor{blue}{\textsf{2}}};
		\node [style=box] (87) at (5, -0.75) {$X(1,1)$};
		\node [style=box] (88) at (2.25, -0.75) {$Z(1,1)$};
		\node [style=2 control] (89) at (2.25, 0.75) {\textcolor{blue}{\textsf{2}}};
		\node [style=box] (90) at (7.75, -0.75) {$Z(1,1)$};
		\node [style=2 control] (91) at (7.75, 0.75) {\textcolor{blue}{\textsf{2}}};
		\node [style=none] (92) at (10, 0) {$=$};
		\node [style=none] (93) at (-11.5, 0.75) {};
		\node [style=none] (94) at (-11.5, -0.75) {};
		\node [style=none] (95) at (-0.75, 0.75) {};
		\node [style=none] (96) at (-0.75, -0.75) {};
		\node [style=none] (97) at (0.75, 0.75) {};
		\node [style=none] (98) at (0.75, -0.75) {};
		\node [style=none] (99) at (9.25, 0.75) {};
		\node [style=none] (100) at (9.25, -0.75) {};
		\node [style=none] (101) at (10.75, 0.75) {};
		\node [style=none] (102) at (10.75, -0.75) {};
		\node [style=none] (103) at (13.25, 0.75) {};
		\node [style=none] (104) at (13.25, -0.75) {};
	\end{pgfonlayer}
	\begin{pgfonlayer}{edgelayer}
		\draw (76) to (77);
		\draw [style=simple] (79) to (78);
		\draw [style=simple] (81) to (80);
		\draw [style=simple] (85) to (84);
		\draw (86) to (87);
		\draw [style=simple] (89) to (88);
		\draw [style=simple] (91) to (90);
		\draw [style=simple] (93.center) to (83);
		\draw [style=simple] (83) to (79);
		\draw [style=simple] (79) to (76);
		\draw [style=simple] (76) to (81);
		\draw [style=simple] (81) to (95.center);
		\draw [style=simple] (94.center) to (78);
		\draw [style=simple] (78) to (77);
		\draw [style=simple] (77) to (80);
		\draw [style=simple] (80) to (96.center);
		\draw [style=simple] (97.center) to (89);
		\draw [style=simple] (89) to (86);
		\draw [style=simple] (86) to (91);
		\draw [style=simple] (91) to (99.center);
		\draw [style=simple] (100.center) to (90);
		\draw [style=simple] (90) to (87);
		\draw [style=simple] (87) to (88);
		\draw [style=simple] (88) to (98.center);
		\draw [style=simple] (101.center) to (85);
		\draw [style=simple] (85) to (103.center);
		\draw [style=simple] (104.center) to (84);
		\draw [style=simple] (84) to (102.center);
	\end{pgfonlayer}
\end{tikzpicture}
    \end{equation}
    To construct the controlled $-H^\dagger$ up to a controlled global phase, we apply Eq.~\eqref{eq:tczwwphase}, conjugating the target by Hadamards for the X rotations per Definition~\ref{def:X-phase-gate}.  As we require three such gates, the controlled global phase becomes $\zeta^{7\cdot 3} = \zeta^3 = \omega$; thus, the necessary correction is the Clifford $S^\dagger$ gate on the control qutrit (i.e. $\ket{2}$-controlled $\omega^2 \mathbb{I}$).
\end{proof}

We can now construct the $R$ gate in Clifford+$T$. 
However, direct substitution of the 24 $T$-count $\ket{2}$-controlled $-H^\dagger$ gate of Lemma~\ref{lemma:tcmhdag} into Eq.~\eqref{eq:R-circuit} yields a $T$-count~$63$ construction. 
We can do better by combining the two iterations of the controlled $-H^\dagger$ in a smarter way.
\begin{theorem}
    The qutrit $R$ gate can be constructed unitarily in Clifford+$T$ with $T$-count~39, provided there is a borrowed (i.e. returned to its starting state) ancilla available.
\end{theorem}
\begin{proof}
    The equality of the circuits below can be verified by direct computation or by noting that it applies $Z(3,3) = \mathbb{I}$ to the target when the control is $\ket{0}$ or $\ket{1}$, and $(H^\dagger)^2 = -\tau_{(12)}$ otherwise.
    \begin{equation}
        \label{eq:tcmd12uptoglobalphase}
        \begin{tikzpicture}
	\begin{pgfonlayer}{nodelayer}
		\node [style=box] (0) at (-6.5, -0.75) {$Z(1,1)$};
		\node [style=none] (1) at (-8, 0.75) {};
		\node [style=2 control] (3) at (-3.25, 0.75) {\textcolor{blue}{\textsf{2}}};
		\node [style=box] (4) at (-3.25, -0.75) {$\zeta^7 X(1,1)$};
		\node [style=2 control] (5) at (0.5, 0.75) {\textcolor{blue}{\textsf{2}}};
		\node [style=box] (6) at (0.5, -0.75) {$\zeta^7 Z(1,1)$};
		\node [style=box] (7) at (3.75, -0.75) {$Z(1,1)$};
		\node [style=2 control] (8) at (7, 0.75) {\textcolor{blue}{\textsf{2}}};
		\node [style=box] (9) at (7, -0.75) {$\zeta^7 X(1,1)$};
		\node [style=box] (10) at (10.25, -0.75) {$Z(1,1)$};
		\node [style=none] (12) at (11.75, 0.75) {};
		\node [style=none] (16) at (-9, 0) {$=$};
		\node [style=none] (17) at (-10, 0.75) {};
		\node [style=none] (18) at (-10, -0.75) {};
		\node [style=none] (20) at (-13.75, -0.75) {};
		\node [style=2 control] (21) at (-11.5, 0.75) {\textcolor{blue}{\textsf{2}}};
		\node [style=box] (22) at (-11.5, -0.75) {$\tau_{(12)}$};
		\node [style=none] (23) at (-8, -0.75) {};
		\node [style=none] (24) at (11.75, -0.75) {};
		\node [style=box] (25) at (-12.75, 0.75) {$R$};
		\node [style=none] (26) at (-13.75, 0.75) {};
		\node [style=box] (27) at (-6.5, 0.75) {$S^\dagger$};
	\end{pgfonlayer}
	\begin{pgfonlayer}{edgelayer}
		\draw [style=simple] (0) to (4);
		\draw [style=simple, in=90, out=-90] (3) to (4);
		\draw [style=simple] (5) to (6);
		\draw [style=simple] (8) to (9);
		\draw [style=simple] (3) to (5);
		\draw [style=simple] (5) to (8);
		\draw [style=simple] (10) to (9);
		\draw [style=simple] (9) to (7);
		\draw [style=simple] (7) to (6);
		\draw [style=simple] (6) to (4);
		\draw [style=simple] (21) to (17.center);
		\draw [style=simple] (20.center) to (22);
		\draw [style=simple] (22) to (18.center);
		\draw [style=simple] (21) to (22);
		\draw (23.center) to (0);
		\draw (10) to (24.center);
		\draw [style=simple] (21) to (25);
		\draw [style=simple] (25) to (26.center);
		\draw [style=simple] (8) to (12.center);
		\draw [style=simple] (1.center) to (27);
		\draw [style=simple] (27) to (3);
	\end{pgfonlayer}
\end{tikzpicture}
    \end{equation}
    To get a circuit for the $R$ gate we simply bring the $\ket{2}$-controlled $\tau_{(12)}$ to the other side (as it is its own inverse). The total $T$-count of the resulting circuit is then $15+3\cdot 8 = 39$.
\end{proof}

As we can construct the $R$ gate as a Clifford+$T$ circuit, any unitary that can be exactly constructed in the Clifford+$R$ gate set can then be exactly (as opposed to approximately) constructed in the Clifford+$T$ gate set. 
Although do note that our conversion presently seems rather inefficient, as the circuit in Eq.~\eqref{eq:R-circuit} requires 39 $T$ gates.
\begin{corollary}
	\label{cor:inclusion}
    The Clifford+$R$ gate set is a subset of the Clifford+$T$ gate set.
\end{corollary}

A natural question to ask now is whether we can do better. Do we
really need two qutrits to write the $R$ gate as a Clifford+$T$
unitary? The answer is \emph{yes}: it is not possible to construct the
$R$ gate using just single-qutrit Clifford+$T$ gates.  This
follows from the normal form that was found for single-qutrit
Clifford+$T$ unitaries in Ref.~\cite{GlaudellA2019canonical}.  Since
the proof of this is rather technical we present the details in
Appendix~\ref{sec:R-not-in-T-proof}, and just give a sketch here.

The group of $3\times 3$ unitary matrices acts on the 8-dimensional
real vector space of traceless Hermitian matrices. This action
defines, for each $3\times 3$ unitary matrix $U$, an $8\times 8$ real
matrix~$\overline{U}$ known as the \emph{adjoint representation} of
$U$. One can then gather information about $U$ by studying its adjoint
representation $\overline{U}$. In particular, it is a consequence of
the normal forms for single-qutrit Clifford+$T$ circuits introduced in Ref.~\cite{GlaudellA2019canonical} 
that the adjoint representation of a
single-qutrit Clifford+$T$ operator has a very specific block matrix
form (see Proposition~\ref{prop:adjointstructure} in
Appendix~\ref{sec:R-not-in-T-proof} below). It can then be shown by
computation that $\overline{R}$ is not of the appropriate form and
therefore not Clifford+$T$ .

Another natural question is the converse to Corollary~\ref{cor:inclusion}: is the $T$ gate included in Clifford+$R$? I.e., is the inclusion of Clifford+$R$ within Clifford+$T$ strict? We will show that this is indeed the case. We begin by considering matrix representations of circuits over Clifford+$R$.

\begin{lemma}
	\label{lemma:ringofgens}
	Let $U$ be a qutrit Clifford+$R$ unitary. Then up to a global phase $U$ has a matrix representation in the computational basis with entries in the number ring $\mathbb{T}[\omega]$.
\end{lemma}
\begin{proof}
By the definitions of $S$, $H$, $R$, and CX and Lemma~\ref{lemma:gatesetring}.
\end{proof}

\begin{proposition}
	\label{prop:tnotinrtwo}
	$T\not\in$ Clifford+$R$
\end{proposition}
\begin{proof}
	We have $T\in$ Clifford+$R$ if there exists a unitary circuit over Clifford+$R$ which performs the operation $T\otimes \mathbb{I}_n$ up to a global phase for some $n\in\mathbb{N}$ where $\mathbb{I}_n$ is the $n$-qutrit identity. 
	In the computational basis, $T\otimes \mathbb{I}$ has a matrix representation with entries from the set $\{0,1,\zeta,\zeta^8\}$. 
	By Lemma~\ref{lemma:ringofgens}, we know that if $T\otimes \mathbb{I}$ permits an exact circuit over Clifford+$R$ we must have $\{0,c,c\zeta,c\zeta^8\}\subset\mathbb{T}[\omega]$ for at least some global phase $c\in\mathbb{C}$ which satisfies $c^*c=1$. 
	As $\mathbb{T}[\omega]$ is closed under conjugation, we then also have $c^*\in \mathbb{T}[\omega]$, and as it is closed under multiplication we then have $c^*c\zeta = \zeta\in\mathbb{T}[\omega]$.
	However, it is well-known that $\zeta\not\in\mathbb{T}[\omega]$, and so there exists no such global phase $c$ (see Appendix~\ref{app:rootofunity}). Hence, no such suitable $c$ exists. As $n$ was arbitrary, we conclude that no Clifford+$R$ circuit exactly implements $T$ in the computational basis.
\end{proof}

\begin{corollary}
	Clifford+$R$ $\subsetneq$ Clifford+$T$.
\end{corollary}

\section{Conclusion}\label{sec:conclusion}
In summary, we showed that the universal fault-tolerant qutrit Clifford+$R$ gate set is a subset of Clifford+$T$, by providing a two-qutrit, $T$-count 39 unitary Clifford+$T$ construction of the $R$ gate.  We prove that our construction is optimal in the number of qutrits by applying the single-qutrit Clifford+$T$ normal form of Glaudell, Ross, and Taylor in Ref.~\cite{GlaudellA2019canonical}.  Moreover, we prove that Clifford+$R$ is a \emph{strict} subset of Clifford+$T$ by showing that regardless of the number of ancillae qutrits, the $T$ gate is impossible to exactly synthesize unitarily in the Clifford+$R$ gate set.

This result is surprising for several reasons.
While several papers have studied the Clifford+$R$ gate set, it was not known that it is a subset of Clifford+$T$, much less a strict subset.  Therefore, we find evidence that the $R$ gate is not more powerful in practice than the $T$ gate, contrary to what was previously believed. 
In fact, we find that Clifford+$T$ is strictly more powerful than Clifford+$R$ for the reason that Clifford+$T$ can exactly synthesize every gate in Clifford+$R$ up to a constant factor of overhead, while the converse is not true.  
The additional gates Clifford+$T$ can represent might be important, as it was for instance conjectured in Ref.~\cite{BocharovA2016ternaryarithmetics} that not all ternary classical reversible gates can be exactly represented in Clifford+$R$, while they can be exactly represented in Clifford+$T$ (a consequence of a result in forthcoming work~\cite{YehL2022qutritcontrols}).
Additionally, our result also means that much of the work done on Clifford+$R$ can now be directly translated to the Clifford+$T$ setting.  For example, the universal approximate synthesis algorithms of~\cite{BocharovA2016topologicalcompilation,BocharovA2016optimalitymetaplectic} can now also be used to synthesise Clifford+$T$ circuits.
One final observation is that our result also demonstrates a way in which qutrit Clifford+$T$ is different from that of qubit Clifford+$T$.
While all the one-qubit Clifford+$T$ circuits that can be constructed with and without ancillae coincide~\cite{KliuchnikovV2013qubitctrlt}, our result shows that this is not true for qutrits, as the single-qutrit $R$ gate cannot be constructed in single-qutrit Clifford+$T$, but can be constructed using one borrowed ancilla.

A natural starting point for future work is to find a lower $T$-count decomposition of the $R$ gate as we have here only attempted preliminary circuit simplification.  
Beyond continued search for more optimal decompositions, an alternate approach to ascertain a lower bound on the $T$-count necessary to prepare the $R$ state would be through leveraging the resource theory of non-stabiliser states, for instance the mana~\cite{VeitchV2014resourcetheorystab} and thauma~\cite{WangX2020boundmsd} measures of magic.
Alternatively, it might be possible to find a normal form for multi-qutrit Clifford+$T$ unitaries which is $T$-optimal,
which would then also give us an optimal decomposition of the $R$ gate.

Our results also pave the way to deriving a full characterisation of which qutrit unitaries can be exactly implemented over the Clifford+$T$ gate set. We conjecture that, as in the qubit case~\cite{GilesB2013multiqubitcliffordplustsynthesis}, any qutrit unitary with entries in $\mathbb{T}[\zeta]$ can be exactly synthesised over Clifford+$T$.

\bibliography{RgateinqutritCliffordplusT}

\appendix

\section{\texorpdfstring{$\zeta\not\in\mathbb{T}[\omega]$}{Primitive ninth roots of unity are not in T[omega]}}
\label{app:rootofunity}
We provide an elementary proof that primitive ninth roots of unity are not elements of $\mathbb{T}[\omega]$. Note that $\zeta\in\mathbb{T}[\omega]$ only if $\zeta\in\mathbb{Q}[\omega]$, where $\mathbb{Q}[\omega]$ is a field. As $\{1,\omega\}$ forms a basis for $\mathbb{Q}[\omega]$ over $\mathbb{Q}$, if $\zeta\in\mathbb{Q}[\omega]$ we would necessarily require some $a,b\in\mathbb{Q}$ such that
\[
\zeta=a+b\omega\implies\zeta^3 = (a+b\omega)^3 = \omega.
\]
Expanding, reducing powers using $1+\omega+\omega^2=0$, and collecting terms, we find
\[
(a^3-3ab^2+b^3) + (3a^2 b-3ab^2)\omega = \omega.
\]
Therefore, by equating coefficients of our basis elements on each side we conclude that we need
\begin{align}
a^3-3ab^2+b^3&=0\label{eq:nulleqn}\\
3a^2 b-3ab^2&=1.\label{eq:oneeqn}
\end{align}
Note that clearly $a,b\neq 0$ if Eq.~\ref{eq:oneeqn} is to be satisfied. Letting $r=a/b\in\mathbb{Q}$, we rearrange Eq.~\ref{eq:nulleqn} and find
\begin{equation}
r^3-3r+1=0.\label{eq:reqn}
\end{equation}
Since $r\neq 0$, let $r=s/t$ for $s,t\in\mathbb{Z}$, $s,t\neq 0$, and $\gcd(s,t)=1$ without loss of generality. Necessarily, we would have
\begin{equation}
\label{eq:intequation}
s^3-3st^2+t^3=0. 
\end{equation}
For any prime $p\mid s$, we clearly have $p\mid t^3$ implying $p\mid t$. Similarly, for any prime $q\mid t$, we must have $q\mid s^3$ and thus $q\mid s$. As we have assumed $\gcd(s,t)=1$, we conclude that no prime can divide $s$ nor $t$ and so $s,t$ must be units in $\mathbb{Z}$ as both are necessarily nonzero. No combination of $s,t =\pm 1$ satisfies Eq.~\ref{eq:intequation}, and thus we conclude no $r\in\mathbb{Q}$ satisfies Eq.~\ref{eq:reqn}. From this, we deduce there are no $a,b\in\mathbb{Q}$ such that
\[
\zeta=a+b\omega
\]
and thus $\zeta\not\in \mathbb{Q}[\omega]\implies\zeta\not\in \mathbb{T}[\omega]$.

\section{The \texorpdfstring{$R$}{R} gate is not a single-qutrit Clifford+\texorpdfstring{$T$}{T} unitary}\label{sec:R-not-in-T-proof}

We start with a set of definitions. These are based on the work done in Ref.~\cite{GlaudellA2019canonical}.

\begin{definition}
    Let $K=\{2^k~|~k\in\mathbb{N}\}$, $\alpha=\sin(2\pi/9)$, and $L=\{\alpha^k~|~k\in\mathbb{N}\}$. We define the following number rings:
    \begin{align*}
        \mathbb{D}&:= K^{-1}\mathbb{Z} = \left\{\frac{a}{2^k}~\middle|~a\in\mathbb{Z}~\text{and}~k\in\mathbb{N}\right\}\\
        \mathbb{D}[\alpha]&=\{a+b\alpha+c\alpha^2+d\alpha^3+e\alpha^4+f\alpha^5~|~a,b,c,d,e,f\in\mathbb{D}\}\\
        \mathbb{A}&:=L^{-1}\mathbb{D}[\alpha] = \left\{\frac{a}{\alpha^k}~\middle|~a\in\mathbb{D}[\alpha]~\text{and}~k\in\mathbb{N}\right\}
    \end{align*}
\end{definition}

Additionally, we will rely on the following quotient ring:

\begin{definition}
    Let $\mathbb{Z}_3:=\mathbb{Z}/(3)$ be the ring of integers modulo 3.
\end{definition}

Using our definitions we can introduce the following ring homomorphism:

\begin{definition}
    Let $\rho:\mathbb{D}[\alpha]\rightarrow\mathbb{Z}_3$ be the ring homomorphism defined by $\rho(q)=q\pmod{\alpha}$ for $q\in\mathbb{D}[\alpha]$. In particular, $\rho(1/2)=2$, $\rho(3)=0$, and $\rho(\alpha)=0$.
\end{definition}

To account for powers of $\alpha$ that appear in the denominator of elements of $\mathbb{A}$, we also introduce the following terminology:

\begin{definition}
    Let $q\in\mathbb{A}$. There always exists some $k\in\mathbb{N}$ for which $\alpha^k q\in\mathbb{D}[\alpha]$. We call $k$ a \emph{denominator exponent} of $q$, and the least such $k$ is called the \emph{least denominator exponent} (LDE). The LDE of a vector or matrix over $\mathbb{A}$ is defined as the largest LDE of their individual elements.
\end{definition}

\begin{definition}
    Let $q\in\mathbb{A}$ and let $k$ be a denominator exponent of $q$. Then the \emph{$k$-residue of $q$}, $\rho_k(q)$ is defined as
    \[
        \rho_k(q) := \rho(\alpha^k q)\in\mathbb{Z}_3.
    \] 
    The $k$-residue of a vector or matrix is defined component-wise.
\end{definition}

The rings we introduced will encompass the entries of Clifford+$T$ matrices in a certain representation called the adjoint representation, which we can describe as follows. Consider the space $\mathbb{H}$ of traceless $3\times 3$ Hermitian
matrices. This space forms an 8-dimensional real vector space and can
be endowed with an inner product by defining $\langle M, M' \rangle
=\Tr(M^\dagger M')$, for any $M,M'\in\mathbb{H}$. As the trace is both cyclic and fixed under transposition of arguments, we have $\langle M,M'\rangle^* = \langle M,M'\rangle$ so that inner product of two traceless Hermitian matrices is necessarily real. It is
straightforward to verify that if $U$ is a $3\times 3$ unitary matrix,
then conjugation by $U$ defines a linear operator on $\mathbb{H}$.

\begin{definition}
    \label{def:adjointrep}
    Let $U$ be a $3\times 3$ unitary matrix. We define the linear
    operator $\overline{U}:\mathbb{H}\to\mathbb{H}$ by
    \(
    \overline{U}(H) = UMU^\dagger
    \)
    for every $M\in\mathbb{H}$. The operator $\overline{U}$ is the
    \emph{adjoint representation} of $U$.
\end{definition}

The adjoint representation $U\mapsto \overline{U}$ defines a group
    homomorphism from $U(3,\mathbb{C})$ to $SO(8,\mathbb{R})$.

For $U$ a Clifford+$T$ operator, we will be interested in the matrix
representation of $\overline{U}$ in some convenient basis. Following
\cite{GlaudellA2019canonical}, for a single-qutrit Pauli $P$, we set
\[
P_{\pm} = \frac{P^\dagger\pm P}{\sqrt{\Tr[(P^\dagger\pm P )^2]}}
\]
in order to define a basis $\mathcal{B}$ for $\mathbb{H}$.

\begin{definition}
    \label{def:adjointbasis}
    Let $X$ and $Z$ be the single-qutrit Pauli operators and let
    $\mathbb{H}$ be the inner product space of $3\times 3$ traceless Hermitian
    matrices. We define the orthogonal basis $\mathcal{B}$ for
    $\mathbb{H}$ as follows
    \[
    \mathcal{B} = \{ Z_{+}, X_{+}, (XZ)_{+}, (XZ^2)_{+},
    Z_{-}, X_{-}, (XZ)_{-}, (XZ^2)_{-}\}.
    \]
\end{definition}

If $U$ is a Clifford+$T$ operator, then the matrix for $\overline{U}$
in the basis $\mathcal{B}$ (ordered as in
Definition~\ref{def:adjointbasis}) has several useful properties, as
detailed in the following proposition, whose proof can be found in
\cite[Remark~4.15, Remark~4.18, and Proposition~4.20]{GlaudellA2019canonical}.

\begin{proposition}
    \label{prop:adjointstructure}
    Let $U$ be a $3\times 3$ unitary matrix and assume that $U$ can be
    exactly represented by an ancilla-free single-qutrit Clifford+$T$
    circuit. Then, in the basis $\mathcal{B}$, the operator
    $\overline{U}$ has entries in the number ring $\mathbb{A}$. Write
    \[
    \overline{U} =
    \left(
    \begin{array}{c|c}
        A & B \\ \hline
        C & D
    \end{array} \right)
    \]
    where $A$, $B$, $C$, and $D$ are $4\times 4$ matrices. If the minimal $T$-count of $U$ restricted to single-qutrit circuits is $k$, then the LDE of submatrix $A$ is $2k$ and the following statements hold:
    \begin{itemize}
        \item If $k=0$, then $U$ is a Clifford operator.
        \item If $k>0$, then up to generalized row and column permutations over $\mathbb{Z}_3$,
        \[
            \rho_{2k}(A) \sim \begin{bmatrix}
                0 & 0 & 0 & 0 \\
                0 & 2 & 2 & 2 \\
                0 & 2 & 2 & 2 \\
                0 & 2 & 2 & 2
            \end{bmatrix}\qquad\text{and}\qquad\rho_{2k+1}(C) \sim \begin{bmatrix}
            0 & 0 & 0 & 0 \\
            0 & 1 & 1 & 1 \\
            0 & 1 & 1 & 1 \\
            0 & 1 & 1 & 1
        \end{bmatrix}.
        \]
    \end{itemize}
\end{proposition}

We are now in a position to prove that $R$ cannot be represented over
the Clifford+$T$ gate set without using ancillae.

\begin{proposition}
    \label{prop:RnotCT}
    The $R$ gate cannot be represented by a single-qutrit ancilla-free
    Clifford+$T$ circuit.
\end{proposition}

\begin{proof}
    Direct computation yields
    \[
    \overline{R} =
    \left(
    \begin{array}{c|c}
        A & B \\ \hline
        C & D
    \end{array} \right)  
    \]
    where
    \[
    A=D = \frac{1}{3}
    \left(
    \begin{array}{cccc}
        3 &  0 &  0 &  0 \\
        0 & -1 &  2 &  2 \\
        0 &  2 & -1 &  2 \\
        0 &  2 &  2 & -1
    \end{array} \right)\qquad\text{and}\qquad B=C=0
    \]
    Thus $\overline{R}$ is a matrix over $\mathbb{A}$. The LDE of $A$ is 6, and thus we compute
    \[
        \rho_6(A) = \begin{bmatrix}
            0 & 0 & 0 & 0 \\
            0 & 2 & 2 & 2 \\
            0 & 2 & 2 & 2 \\
            0 & 2 & 2 & 2
        \end{bmatrix}\qquad \text{and}\qquad\rho_7(C) =  \begin{bmatrix}
        0 & 0 & 0 & 0 \\
        0 & 0 & 0 & 0 \\
        0 & 0 & 0 & 0 \\
        0 & 0 & 0 & 0
    \end{bmatrix}.
    \]
    In particular, $\rho_7(C)$ is \emph{not} equivalent up to generalized row/column permutations to the matrix
    \[
        \begin{bmatrix}
            0 & 0 & 0 & 0 \\
            0 & 1 & 1 & 1 \\
            0 & 1 & 1 & 1 \\
            0 & 1 & 1 & 1
        \end{bmatrix}.
    \]
    Thus, $R$ cannot be represented by a single-qutrit
    ancilla-free Clifford+$T$ circuit.
\end{proof}

\end{document}